\documentclass[12pt, a4paper]{article}

\usepackage{amsmath}
\usepackage{amsfonts}
\usepackage{amssymb}
\usepackage{amsthm}
\usepackage[top=2cm, bottom=2cm, left=2.5cm, right=2.5cm]{geometry}
\usepackage{ifthen}
\usepackage{needspace}
\usepackage{parskip}
\usepackage{stmaryrd}
\usepackage[]{todonotes}
\usepackage{xspace}

\allowdisplaybreaks
\usepackage[utf8]{inputenc}

\renewenvironment{proof}[1][\proofname]{{\bfseries #1: }}{\qed}

\newtheoremstyle{defstyle}{10pt}{5pt}{\addtolength{\leftskip}{1\leftmargini}\addtolength{\rightskip}{1\leftmargini}}{-1\leftmargini}{\small\scshape\bfseries}{:}{\newline}{#1 #2\ifthenelse {\equal {#3}{}} {}{ (\textnormal{\textsc{#3}})}}

\newtheoremstyle{thmstyle}{10pt}{5pt}{\addtolength{\leftskip}{1\leftmargini}\addtolength{\rightskip}{1\leftmargini} \slshape}{-1\leftmargini}{\small\scshape\bfseries}{:}{\newline}{#1 #2\ifthenelse {\equal {#3}{}} {}{ (\textnormal{\textsc{#3}})}}

\theoremstyle{defstyle}

\newtheorem{mydef}{Definition}

\theoremstyle{thmstyle}

\newtheorem{mythm}{Theorem}
\newtheorem{mylem}{Lemma}
\newtheorem{myrem}{Remark}
\newtheorem{mycor}{Corollary}

\DeclareFontFamily{OT1}{pzc}{}
\DeclareFontShape{OT1}{pzc}{m}{it}%
              {<-> s * [1.15] pzcmi7t}{}
\DeclareMathAlphabet{\mathpzc}{OT1}{pzc}%
                                 {m}{it}

\allowdisplaybreaks

\usepackage{alltt}
\usepackage{amsmath}
\usepackage{amssymb}
\usepackage{stmaryrd}
\usepackage{graphicx}
\usepackage{tikz}
\usepackage{wrapfig}
\usepackage{xspace}

\newcommand{\soprogs}{\ensuremath{\textnormal{\textsf{Prog}}}\xspace}
\newcommand{\sonprogs}{\ensuremath{\textnormal{\textsf{ordProg}}}\xspace}
\newcommand{\sovars}{\ensuremath{\textnormal{\textsf{Var}}}\xspace}
\newcommand{\ttt}[1]{\textnormal{\texttt{#1}}}

\newcommand{\T}{\ensuremath{\textnormal T}\xspace}
\newcommand{\Tp}{\ensuremath{\T_{\mathit{prob}}}\xspace}
\newcommand{\Tps}{\ensuremath{\Tp^*}\xspace}
\newcommand{\Exp}[2]{\ensuremath{\textnormal E_{#1}(#2)}\xspace}
\renewcommand{\Pr}{\textnormal{Pr}}

\newcommand{\Problem}[1]{\mathcal{#1}}
\newcommand{\cProblem}[1]{\overline{\text{\footnotesize{$\Problem{#1}$}}}}

\newcommand{\bbbn}{\mathbb{N}\xspace}
\newcommand{\bbbq}{\mathbb{Q}\xspace}
\newcommand{\bbbs}{\mathbb{S}\xspace}

\newcommand{\UHP}{\Problem{UH}\xspace}
\newcommand{\cUHP}{\cProblem{UH}\xspace}
\newcommand{\AST}{\ensuremath{\Problem{AST}}\xspace}

\newcommand{\EXP}{\ensuremath{\Problem{EXP}}}
\newcommand{\LEXP}{\ensuremath{\Problem{LEXP}}}
\newcommand{\REXP}{\ensuremath{\Problem{UEXP}}}

\newcommand{\leqm}{\mathrel{\:{\leq}_{\textnormal{m}}}}

\begin{document}

\title{Analyzing Expected Outcomes and Almost--Sure Termination of Probabilistic Programs is Hard\footnote{This research is funded by the Excellence Initiative of the German federal and state governments and by the EU FP7 MEALS project.}}
\author{Benjamin Lucien Kaminski \and Joost-Pieter Katoen}
\date{\today}
\maketitle

\begin{abstract}
This paper considers the computational hardness of computing expected outcomes and deciding almost--sure termination of probabilistic programs.
We show that deciding almost--sure termination and deciding whether the expected outcome of a program equals a given rational value is $\Pi^0_2$--complete.
Computing lower and upper bounds on the expected outcome is shown to be recursively enumerable and $\Sigma^0_2$--complete, respectively.
 \end{abstract}
 
\allowdisplaybreaks

\section{Introduction}

Probabilistic programs~\cite{DBLP:journals/jcss/Kozen81} are imperative sequential programs with the ability to draw values at random from probability distributions.
They are used in security to describe cryptographic constructions (such as randomized encryption) and security experiments~\cite{DBLP:journals/toplas/BartheKOB13}, in machine learning to describe distribution functions that are analyzed using Bayesian inference~\cite{DBLP:journals/corr/BorgstromGGMG13}, and in randomized algorithms.
They are typically just a small number of lines, but hard to understand and analyze, let alone algorithmically.

This paper considers a precise classification of the computational hardness of solving two main analysis problems for probabilistic programs: (1) almost--sure termination \cite{pnueli} --- does a program terminate with probability one? --- and (2) computing expected outcomes --- is the expected outcome of a program (variable) equal, smaller, or larger than a given rational number? Expected outcomes correspond to McIver \& Morgan's weakest pre--expectation semantics of pGCL, the probabilistic version of Dijkstra's guarded command language~\cite{mciver}.

Several results on computational hardness in connection with analyzing probabilistic programs have been reported in the literature, like non--recursive--enumerability results for probabilistic rewriting logic \cite{bournez} and decidability results for restricted probabilistic programming languages \cite{murawski}.
A lot of work has also been done towards automated reasoning for almost--sure termination.
For instance \cite{sneyers} gives an overview of some particularly interesting examples of probabilistic logical programs and the according intuition for proving almost--sure termination. 
Arons \emph{et al.} reduce almost--sure termination to termination of a non--deterministic program by means of a planner \cite{arons2003parameterized}.
In \cite{esparza}, a pattern--based approach which exploits this idea together with a prototypical tool support is presented. 

Despite several approaches to tackle the problem of almost--sure termination in an automated manner, the majority of the literature does not consider the hardness of the problem or states that it must intuitively be harder to solve than the termination problem for ordinary, i.e., non--probabilistic programs.
For instance in \cite{onlymorgan} it is noted that while partial correctness for small--scale examples is not harder to prove than for ordinary programs, the case for total correctness of a probabilistic loop must be harder to analyze.
As another example \cite{esparza} suggests that almost--sure termination must be harder to decide than ordinary termination since for the latter a topological argument suffices while for the former arithmetical reasoning is needed. 
 
Aside from the intuition that almost--sure termination must be somewhat harder to decide, to the best of our knowledge, there seems to be yet no \emph{precise} classification of the computational hardness of deciding this major analysis problem.
This gap is bridged by this paper.
Such a precise classification is not only of theoretical interest, but allows for deeper insights into the specific difficulties of dealing with the problem of almost--sure termination and may help in identifying subclasses of programs for which it becomes easier so solve. 
Through our classification we cannot only establish that almost sure--termination is in fact strictly harder to decide than ordinary termination but we can also make a statement on the upper bound of the hardness of the problem.

In this paper we study and formalize the problem of computing expected outcomes and the problem of deciding almost--sure termination, and we establish the following hardness results:
We first show that computing lower bounds on the expected outcome of program variable $v$ by executing a probabilistic program $P$ is recursively enumerable.
Computing upper bounds for the expected outcome is shown to be $\Sigma_2^0$--complete, whereas deciding whether the expected outcome of $v$ equals some rational is shown to be $\Pi_2^0$--complete.
Finally, almost--sure termination of $P$ is shown to be $\Pi_2^0$--complete.
The immediate consequences of the latter result are twofold: (1) deciding almost--sure termination is not only intuitively but provably strictly harder than deciding termination for ordinary programs, and (2) deciding almost--sure termination is \emph{not} harder than deciding whether an \emph{ordinary} program halts on all or infinitely many inputs \cite{odifreddi2}.

Regarding the consequences of the $\Sigma_2^0$--completeness of  computing upper bounds for expected outcomes we note the following: 
If both upper and lower bounds for expected outcomes were recursively enumerable, then  each expected outcome would be a computable real number. 
However, since upper bounds are shown to be not recursively enumerable (implied by the $\Sigma_2^0$--completeness), we can establish that in general it is not possible to approximate expected outcomes from above and from below with arbitrary precision.

Our hardness results are established by reductions from the universal halting problem for ordinary programs.
Remarkably the probabilistic programs we use in our reduction obey a certain syntactic schema, namely that the randomization and the actual computation are strictly separated.
We interpret this as possible evidence for the existence of a ``normal form" for probabilistic programs.

\section{Preliminaries}

In order to have a frame of reference in which we can classify the hardness of calculating expected outcomes of probabilistic programs, we first briefly recall the concept of the arithmetical hierarchy:

\begin{mydef}[Arithmetical Hierarchy \cite{kleeneNF,odifreddi1}]
\label{remarithmetic}
The \textbf{class $\boldsymbol{\Sigma_n^0}$} is defined as
\begin{align*}
\Sigma_n^0 ~=~ \Big\{ \Problem A ~\Big|~ &\Problem A = \big\{ \vec x ~\big|~ \exists y_1\, \forall y_2\, \exists y_3\, \cdots\, \exists/\forall y_n\colon~ (\vec x,\, y_1,\, y_2,\, y_3,\, \ldots,\, y_n) \in \Problem R\big\},\\
&\Problem R \textnormal{ is a decidable relation}  \Big\},
\end{align*}
the \textbf{class $\boldsymbol{\Pi_n^0}$} is defined as
\begin{align*}
\Pi_n^0 ~=~ \Big\{ \Problem A ~\Big|~ &\Problem A = \big\{ \vec x ~\big|~ \forall y_1\, \exists y_2\, \forall y_3\, \cdots\, \exists/\forall y_n\colon~ (\vec x,\, y_1,\, y_2,\, y_3,\, \ldots,\, y_n) \in \Problem R \big\},\\&\Problem R \textnormal{ is a decidable relation} \Big\},
\end{align*}
and the \textbf{class $\boldsymbol{\Delta_n^0}$} is defined as $\Delta_n^0 = \Sigma_n^0 \cap \Pi_n^0$, for every $n\in \bbbn$.

Note that we implicitly always quantify over $\bbbq^+$ and that by the $\vec x$'s we mean tuples over $\bbbq^+$.
\emph{Multiple consecutive quantifiers} \emph{of the same type} can be contracted to \emph{one} quantifier of that type, so the number $n$ really refers to the number of necessary \emph{quantifier alternations} rather than to the number of quantifiers used.

A set $\Problem A$ is called \textbf{arithmetical}, iff $\Problem A \in \Gamma_n^0$, for some $\Gamma \in \{\Sigma,\, \Pi,\, \Delta\}$ and some $n \in \mathbb N$.
The inclusion diagram
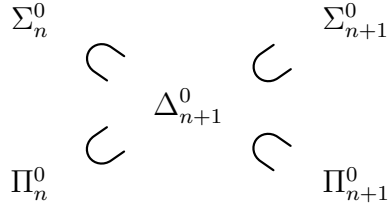
\begin{align*}
\begin{array}{r}
\Sigma_n^0\\\\\\\\
\Pi_n^0\vphantom{\Big(}
\end{array} \begin{array}{c}\mathrel{\rotatebox{-35}{\text{\LARGE $\subset$}}}\\\\\mathrel{\rotatebox{+35}{\text{\LARGE $\subset$}}}\end{array}~ \Delta_{n+1}^0 \,\begin{array}{c}\mathrel{\rotatebox{+35}{\text{\LARGE $\subset$}}}\\\vspace{-0.6em}\\\mathrel{\rotatebox{-35}{\text{\LARGE $\subset$}}}\end{array} \begin{array}{l}
\Sigma_{n+1}^0\\\\\\\\
\Pi_{n+1}^0\vphantom{\Big(}
\end{array}
\end{align*}
holds for every $n \geq 1$, thus the arithmetical sets form a strict hierarchy.
Furthermore note that $\Sigma_0^0 = \Pi_0^0 = \Delta_0^0 = \Delta_1^0$ is exactly the class of the decidable sets and $\Sigma_1^0$ is exactly the class of the recursively enumerable sets.
\end{mydef}
Next we recall the concept of many--one reducibility and the concept of completeness. 
Both these notions allow us to precisely classify the hardness of calculating expected outcomes and deciding almost--sure termination.
\begin{mydef}[Many--One Reducibility \cite{odifreddi1,post44}]
Let $\Problem A,\, \Problem B$ be arithmetical and let $X$ be some appropriate universe, such that $\Problem A,\Problem B \subseteq X$.
A set $\Problem A$ is called \textbf{many--one--reducible} to a set $\Problem B$, denoted $\boldsymbol{\Problem A \leq_\textnormal{\textbf{m}} \Problem B}$, iff there exists a computable function $f\colon X \rightarrow X$, such that
\begin{align*}
\forall \vec x \in X\colon \big(\vec x \in \Problem A \Longleftrightarrow f(\vec x) \in \Problem B\big)~.
\end{align*}
If $f$ is a function, such that $f$ many--one reduces $\Problem A$ to $\Problem B$, we denote this by $\boldsymbol{f\colon \Problem A \leqm \Problem B}$.
Note that $\leqm$ is obviously transitive. 
\end{mydef}
\begin{mydef}[$\boldsymbol{\Gamma_n^0}$--Completeness \cite{odifreddi1}]
A set $\Problem A$ is called \textbf{$\boldsymbol{\Gamma_n^0}$--complete}, for $\Gamma \in \{\Sigma,\, \Pi,\, \Delta\}$, iff both $\Problem A \in \Gamma_n^0$ and $\Problem A$ is \textbf{$\boldsymbol{\Gamma_n^0}$--hard}, meaning $\Problem B \leqm \Problem A$, for any set $\Problem B \in \Gamma_n^0$.
\end{mydef}
An important fact about $\Sigma_n^0$-- and $\Pi_n^0$--complete sets is that they are in some sense the most complicated sets in $\Sigma_n^0$ and $\Pi_n^0$, respectively.
Formally, this can be expressed as follows:
\begin{mylem}[Properties of Complete Sets \cite{davis}]
\label{lemcompleteness}
If $\Problem A$ is $\Sigma_n^0$--complete, then $\Problem{A} \in \Sigma_n^0\setminus \Pi_n^0$. Analogously if $\Problem A$ is $\Pi_n^0$--complete, then $\Problem{A} \in \Pi_n^0\setminus \Sigma_n^0$.
\end{mylem}
Lemma \ref{lemcompleteness} implies in particular that for a $\Sigma_n^0$--complete set $\Problem A$ it holds that $\Problem A \not\in \Delta_n^0$.

\section{Probabilistic Programs}

In order to speak about probabilistic programs and the computations performed by such programs, we first briefly introduce their syntax and their semantics:
\begin{mydef}[Syntax of Probabilistic Programs]
Let $\boldsymbol{\textsf{\textbf{Var}}}$ be the \textbf{set of program variables}. 
The \textbf{set $\boldsymbol{\textsf{Prog}}$ of probabilistic programs} is defined inductively as follows: 
For any $v \in \sovars$ and any arithmetical expression $e$ over \sovars (not to be confused with the sets from the arithmetical hierarchy), the assignment $v \ttt{ := } e$ is in \soprogs. Furthermore if $P_1, P_2 \in \soprogs$, $p \in [0,\, 1] \subseteq \bbbq$, and if $b$ is a Boolean expression over arithmetic expressions then the concatenation $P_1\ttt{;}\: P_2$, the probabilistic choice $\ttt{\{}P_1\ttt{\}} \:[p]\: \ttt{\{}P_2\ttt{\}}$, and the while--loop $\ttt{WHILE}\:\ttt{(} $b$ \ttt{)}\:\ttt{\{} P_1 \ttt{\}}$ are also in \soprogs. 
We call the set of programs that \emph{do not} contain any probabilistic choices the \textbf{set of ordinary programs} and denote this set by \textbf{\textsf{ordProg}}.
\end{mydef}
This syntax is a subset of pGCL originated from McIver  and Morgan \cite{mciver}.
We omitted \texttt{skip}--, \texttt{abort}--, and \texttt{if}--statements, as those are syntactic sugar. 
Furthermore, we do not consider (non--probabilistic) non--determinism.
The operational semantics for our programs is given below:
\begin{mydef}[Semantics of Probabilistic Programs]
Let the set of variable valuations be denoted by $\mathbb V = \{\eta ~|~ \eta\colon \sovars \rightarrow \bbbq^+\}$, let the set of program states be denoted by $\bbbs = \big(\soprogs \cup \{{\downarrow}\}\big) \times \mathbb V \times I \times \{L,\, R\}^*$, for $I = [0,\, 1] \subseteq \bbbq^+$, let $\llbracket e \rrbracket_\eta$ be the evaluation of the arithmetic expression $e$ given the variable valuation $\eta$, and analogously let $\llbracket b \rrbracket_\eta$ be the evaluation of the Boolean expression $b$. 
Then the \textbf{semantics of probabilistic programs} is given by the smallest relation ${\vdash} \subseteq \bbbs \times \bbbs$ which satisfies the following inference rules:
\normalsize\begin{align*}
 (\textnormal{assign})&\frac{}{\langle v\textnormal{\texttt{ := }} e,\, \eta,\, a,\, \theta\rangle ~\vdash~  \langle {\downarrow},\, \eta[v \mapsto \max \{\llbracket e \rrbracket_\eta,\, 0\}],\, a,\, \theta\rangle}\\[0.5\baselineskip]
(\textnormal{concat1})&\frac{\langle P_1,\, \eta,\, a,\, \theta\rangle ~\vdash~  \langle P_1',\, \eta',\, a',\, \theta'\rangle}{\langle P_1\textnormal{\texttt{;}}\:P_2,\, \eta,\, a,\, \theta\rangle ~\vdash~  \langle P_1'\textnormal{\texttt{;}}\: P_2,\, \eta',\, a',\, \theta'\rangle}\\[0.5\baselineskip]
(\textnormal{concat2})&\frac{}{\langle {\downarrow}\textnormal{\texttt{;}}\:P_2,\, \eta,\, a,\, \theta\rangle ~\vdash~  \langle P_2,\, \eta,\, a,\, \theta\rangle}\\[0.5\baselineskip]
(\textnormal{prob1})&\frac{}{\langle \{P_1\}\:[p]\:\{P_2\},\, \eta,\, a,\, \theta\rangle ~\vdash~  \langle P_1,\, \eta,\, a \cdot p,\, \theta \cdot L\rangle}\\[0.5\baselineskip]
(\textnormal{prob2})&\frac{}{\langle \{P_1\}\:[p]\:\{P_2\},\, \eta,\, a,\, \theta\rangle ~\vdash~  \langle P_2,\, \eta,\, a \cdot (1 - p),\, \theta \cdot R\rangle}\\[0.5\baselineskip]
(\textnormal{while1})&\frac{\llbracket b \rrbracket_\eta = \textnormal{True}}{\langle \textnormal{\texttt{WHILE}}\:\ttt{(} b \textnormal{\texttt{)}}\:\texttt{\{} P \textnormal{\texttt{\}}},\, \eta,\, a,\, \theta\rangle ~\vdash~  \langle P\textnormal{\texttt{;}}\: \ttt{WHILE}\:\ttt{(} b \textnormal{\texttt{)}}\:\ttt{\{} P \textnormal{\texttt{\}}},\, \eta,\, a,\, \theta\rangle}\\[0.5\baselineskip]
(\textnormal{while2})&\frac{\llbracket b \rrbracket_\eta = \textnormal{False}}{\langle \textnormal{\texttt{WHILE}}\:\ttt{(} b \textnormal{\texttt{)}}\:\ttt{\{} P \textnormal{\texttt{\}}},\, \eta,\, a,\, \theta\rangle ~\vdash~  \langle {\downarrow},\, \eta,\, a,\, \theta\rangle}
\end{align*}\normalsize
We use $\sigma  \vdash^k \tau$ and $\sigma  \vdash^* \tau$ in the usual sense.
Furthermore we write $\sigma \vdash_{(\textnormal{name})} \tau$ if $\tau$ is inferred by the use of the \textnormal{(name)}--rule (for $\textnormal{name} \in \{\textnormal{assign}$, concat1, $\ldots\})$.
\end{mydef}
The semantics is mostly straightforward except for two things: in addition to the program that is to be executed next and the current variable valuation, each state also stores a string $\theta \in \{L,\ R\}^*$ that indicates which probabilistic choices were made in the past ($\boldsymbol L$eft or $\boldsymbol R$ight) as well as the probability $a$ with which those choices were made.
The graph that is spanned by the $\vdash$--relation can be seen as an unfolding of the MDP--semantics for pGCL provided by Gretz \emph{et al.} \cite{gretz}.

\section{Expected Outcomes and Termination Probabilities}

In this section we formally define the notions of the expected outcome of $P \in \soprogs$ as well as its termination probability. We start by some auxiliary notions:
It is a well--known result due to Kleene that for any program state $\sigma$ for which the successor is inferred without the use of the (prob1)-- or the (prob2)--rule, i.e., the next instruction to be executed is \emph{not}  a probabilistic choice, the successor of $\sigma$ is unique and computable:
\begin{mythm}[The State Successor Function T \cite{kleeneNF}]
\label{remT}
Let $\bbbs_o$ be the set of program states for which the successor is inferred without the use of the (prob1)-- or the (prob2)--rule. Then there exists a total computable function $\T\colon \bbbs_o \rightarrow \big(\bbbs \cup \{\top\}\big) $, such that
\abovedisplayskip=0pt\begin{align*}
\T(\sigma) ~=~ \begin{cases}
\tau, &\textnormal{if } \sigma \vdash \tau\\
\top, &\textnormal{if } \sigma = \langle {\downarrow},\, \eta,\, a,\, \theta \rangle.
\end{cases}
\end{align*}
\end{mythm}
The successor of a state $\sigma \in \bbbs \setminus \bbbs_o$ is not unique, because the program chooses a left or a right branch with some probability. However, if we resolve the probabilistic choice by providing a symbol $L$ or $R$ that indicates whether the left or the right branch shall be chosen, we can come up with a computable function $\Tp$ which computes a unique successor: 
\begin{mycor}[The Probabilistic State Successor Function $\Tp$]
\label{propTp}
There exists a total computable function $\Tp\colon (\bbbs \setminus \bbbs_o) \times \{L,\, R\} \rightarrow \bbbs$, such that
\begin{align*}
\Tp(\sigma,\, s) ~=~ &\begin{cases}
\tau_L,&\textnormal{if }s = L \textnormal{ and } \sigma \vdash_{\textnormal{(prob1)}}\tau_L\\
\tau_R,&\textnormal{if }s = R \textnormal{ and } \sigma \vdash_{\textnormal{(prob2)}}\tau_R.
\end{cases}
\end{align*}
\end{mycor}
While $\T$ and $\Tp$ each compute only the next successor state, we can also define a computable function $\Tps$ that computes the $k$-th successor state according to some sequence $w \in \{L,\, R\}^*$ which tells $\Tps$ how to resolve the probabilistic choices that occur along the computation:
\begin{mycor}[The $k$-th State Successor Function $\Tps$]
\label{corTps}
There exists a total computable function $\Tps\colon \bbbs \times \bbbn \times \{L,\, R\}^* \rightarrow \big(\bbbs \cup \{\top\}\big)$, such that
\begin{align*}
\Tps(\sigma,\, k,\, w) ~=~ \begin{cases}
\tau,&\textnormal{if }\sigma = \langle P,\, \eta,\, a,\, \theta \big\rangle \vdash^k \langle P',\, \eta',\, a',\, \theta \cdot w\big\rangle = \tau\\
\top,&\textnormal{otherwise.}
\end{cases}
\end{align*}
\end{mycor}
So \Tps returns a successor state $\tau$, if $\sigma \vdash^k \tau$, whereupon exactly $|w|$ inferences must use the \textnormal{(prob1)}-- or the \textnormal{(prob2)}--rule and the probabilistic choices are resolved according to $w$. Otherwise \Tps returns $\top$.
Note in particular that for both the inference of a terminal state $\langle{\downarrow},\, \eta,\, a,\, \theta\rangle$ within less than $k$ steps as well as the inference of a $k$-th successor state through less than $|w|$ probabilistic choices, the calculation of \Tps will result in $\top$.

In addition to $\Tps$, we will need three more computable operations for defining the expected outcomes and almost--sure termination:
\begin{mycor}
There exist total computable functions $\alpha\colon \big(\bbbs \cup \{\top\}\big) \rightarrow \bbbq^+$, $\wp\colon \big(\bbbs \cup \{\top\}\big) \times \sovars  \rightarrow \bbbq^+$, and $h\colon \mathbb N \rightarrow \{L,\, R\}^*$, such that
\begin{align}
\alpha(\sigma) ~~&{=}~\begin{cases}a,&\textnormal{if } \sigma = \langle {\downarrow},\, \eta,\,  a,\, \theta\rangle\\0,&\textnormal{otherwise,}\end{cases}\\[.5\baselineskip]
\wp(\sigma,\,  v) ~~&{=}~\begin{cases}\eta(v) \cdot a,&\textnormal{if } \sigma = \langle {\downarrow},\, \eta,\,  a,\, \theta\rangle\\0,&\textnormal{otherwise,}\end{cases}\\[.25\baselineskip]
h~&\textnormal{is a computable bijection.}
\end{align}
\end{mycor}
So the function $\alpha$ takes a state $\sigma$ and returns the probability of reaching $\sigma$, and the function $\wp$ takes a state $\sigma$ and a variable $v$ and returns the probability of reaching $\sigma$ multiplied with the value of $v$ in the state $\sigma$.
Both functions do that only if the provided state $\sigma$ is a terminal state. 
Otherwise they return 0.

We now have all the concepts and notations available for defining the expected outcome of a program variable after executing a probabilistic program and for defining the program's termination probability:
\begin{mydef}[Expected Outcomes and Termination Probabilities]
\label{exp}
Let $P \in \soprogs$ and let $\eta_0 \in \mathbb V$, such that $\forall x \in \sovars\colon \eta_0(x) = 0$.
Then starting in $\eta_0$,
\begin{enumerate}
\item the \textbf{expected outcome} of $v \in \sovars$ after executing $P$, denoted $\boldsymbol{\textnormal{\textbf{E}}_P(v)}$, is given by
\begin{align*}
\Exp{P}{v} ~:=~ \sum_{i \in \bbbn} \: \sum_{j \in \bbbn} \wp\Big(\Tps\big(\langle P,\, \eta_0,\, 1,\, \varepsilon\rangle,\, j,\, h(i)\big),\, v \Big)~,
\end{align*}\normalsize
\item the \textbf{probability that $\boldsymbol P$ terminates}, denoted $\boldsymbol{\textnormal{\textbf{Pr}}_P({\downarrow})}$, is given by
\begin{align*}
\Pr_P({\downarrow}) ~:=~ \sum_{i \in \bbbn} \sum_{j \in \bbbn} \alpha\Big(\Tps\big(\langle P,\, \eta_0,\, 1,\, \varepsilon\rangle,\, j,\, h(i)\big) \Big)~.
\end{align*}
\end{enumerate}
\end{mydef}
$\Exp{P}{v}$ is basically equivalent to the expected reward introduced by Gretz \emph{et al.}\ in \cite{gretz} and thereby coincides with the expectation transformer semantics by McIver and Morgan \cite{mciver}.
The main difference is that the expected reward is defined on a Markov decision process, whereas $\Exp{P}{v}$ is defined on its unfolding.
In principle, for $\Exp{P}{v}$ we sum over all possible numbers of inference steps and over all possible sequences for resolving probabilistic choices, using $\wp$ we filter out the terminal states $\sigma$, and finally sum up the values of $\wp(\sigma,\, v)$.

For $\Pr_P({\downarrow})$ we basically do the same thing but we merely sum up the probabilities of reaching final states, thus ignoring the variable valuations, by using $\alpha$ instead of $\wp$. 
Regarding the termination probability of a probabilistic program, the case of almost--sure termination is of special interest:
\begin{mydef}[Almost--Sure Termination]
We say that a program $P$ \textbf{terminates almost--surely} iff $\Pr_P({\downarrow}) = 1$. Consequently, we define the according set $\boldsymbol{\AST} \subset \soprogs$ as 
\begin{align*}
{P \in \AST} ~~\mathrel{:\Longleftrightarrow}~~ {\Pr_P({\downarrow}) = 1}~.
\end{align*}
\end{mydef}
In order to investigate the complexity of calculating $\Exp{P}{v}$, we define three sets: 
$\LEXP$, which relates to the set of lower bounds of $\Exp{P}{v}$, $\REXP$, which relates to the set of upper bounds of $\Exp{P}{v}$, and $\EXP$ which relates to $\Exp{P}{v}$ itself:
\begin{mydef}[$\LEXP$, $\REXP$, and $\EXP$]
The sets $\boldsymbol \LEXP, \boldsymbol \REXP, \boldsymbol \EXP \subseteq \soprogs \times \sovars \times \bbbq$ are defined as follows:
\begin{align*}
(P,\, v,\, q) \in \LEXP ~~&:\Longleftrightarrow~~ q < \Exp{P}{v}\\
(P,\, v,\, q) \in \REXP ~~&:\Longleftrightarrow~~ q > \Exp{P}{v}\\
(P,\, v,\, q) \in \EXP ~~&:\Longleftrightarrow~~  q = \Exp{P}{v}
\end{align*}
\end{mydef}
\section{Hardness of Computing Expected Outcomes}

We now have all definitions available to begin the investigation of the computational hardness of computing expected outcomes.
The first fact we observe is that lower bounds for expected outcomes are recursively enumerable:
\begin{mylem}
\label{LinSigma}
$\LEXP \in \Sigma_1^0$, thus $\LEXP$ is recursively enumerable.
\end{mylem}
\begin{proof}
\abovedisplayskip=-1\baselineskip\belowdisplayskip=-1\baselineskip
\begin{align*}
~&(P,\, v,\,  q) \in \LEXP\\
\Longleftrightarrow~& q < \Exp{P}{v}\\
\Longleftrightarrow~&q < \sum_{i \in \bbbn} \sum_{j \in \bbbn} \wp\Big(\Tps\big(\langle P,\, \eta_0,\, 1,\, \varepsilon\rangle,\, j,\, h(i)\big),\, v \Big)\\
\Longleftrightarrow~&\exists y_1 \exists y_2\colon q < \sum_{i = 0}^{y_1} \sum_{j = 0}^{y_2} \wp\Big(\Tps\big(\langle P,\, \eta_0,\, 1,\, \varepsilon\rangle,\, j,\, h(i)\big),\, v \Big)\\
\Longrightarrow~&\LEXP \in \Sigma_1^0
\end{align*}
\end{proof}

Recursive enumerability of $\LEXP$ means that lower bounds for expected outcomes can be effectively enumerated by some algorithm.

Now, if the set of upper bounds, i.e., $\REXP$, was recursively enumerable as well, then expected outcomes would be computable reals.
However, the contrary will be established over the course of the following lemmas:
\begin{mylem}
\label{RinSigma2}
$\REXP \in \Sigma_2^0$~.
\end{mylem}
\begin{proof}
\abovedisplayskip=-1\baselineskip\belowdisplayskip=-1\baselineskip
\begin{align*}
~&(P,\, v,\,  q) \in \REXP\\
\Longleftrightarrow~&q > \Exp{P}{v}\\
\Longleftrightarrow~&q > \sum_{i \in \bbbn} \sum_{j \in \bbbn} \wp\Big(\Tps\big(\langle P,\, \eta_0,\, 1,\, \varepsilon\rangle,\, j,\, h(i)\big),\, v \Big)\\
\Longleftrightarrow~&\exists \delta > 0 \:\forall y_1 \:\forall y_2\colon q - \delta > \sum_{i = 0}^{y_1} \sum_{j = 0}^{y_2} \wp\Big(\Tps\big(\langle P,\, \eta_0,\, 1,\, \varepsilon\rangle,\, j,\, h(i)\big),\, v \Big)\\
\Longrightarrow~&\REXP \in \Sigma_2^0
\end{align*}
\end{proof}

Figure \ref{fig} shows a schematic depiction of the intuition behind the formulae defining $\LEXP$ and $\REXP$, respectively.

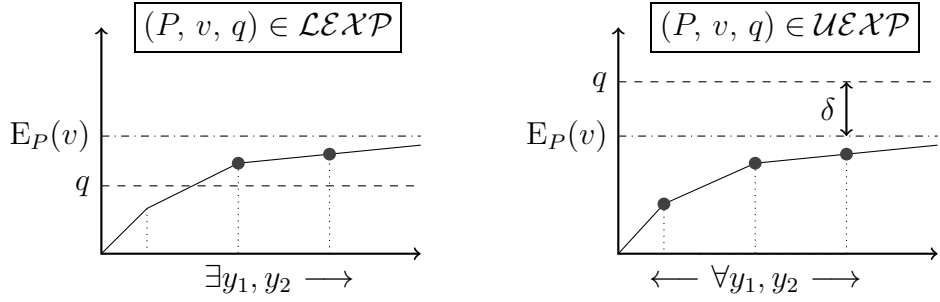
\begin{figure}[t]
\begin{center}
\begin{tikzpicture}[scale=1.2]
    \draw [<->,thick] (0,2.5) node (yaxis) [above] {}
        |- (3.5,0) node (xaxis) [right] {};
    \draw (0,0) coordinate (a_1) -- (0.5,.5) coordinate (a_2);
    \draw (.5,.5) coordinate (a_3) -- (1.5,1) coordinate (a_4);
    \draw (1.5,1) coordinate (a_5) -- (3.5,1.2) coordinate (a_6);
    \draw[dashed] (0,.75) coordinate (b_1) -- (3.5,.75) coordinate (b_2) node[left] at (0,.75) {$q$};
    \coordinate (c1) at (0.5, 0);
    \coordinate (c2) at (0.5, 1);
    \coordinate (c) at (intersection of a_1--a_2 and c1--c2);
    \draw[dotted] (c) -- (xaxis -| c);
    \coordinate (cc1) at (1.5, 0);
    \coordinate (cc2) at (1.5, 1);
    \coordinate (cc) at (intersection of a_3--a_4 and cc1--cc2);
    \draw[dotted] (cc) -- (xaxis -| cc) node[below] at (1.95, 0){${\exists y_1,y_2}~ {\longrightarrow}$};
    \coordinate (ccc1) at (2.5, 0);
    \coordinate (ccc2) at (2.5, 1);
    \coordinate (ccc) at (intersection of a_5--a_6 and ccc1--ccc2);
    \draw[dotted] (xaxis -| ccc) -- (ccc);
    
    \draw[dashdotted] (0, 1.3) coordinate (eeeee1) -- (3.5, 1.3) coordinate (eeeee2);
    \node[left] at (0, 1.3) {$\Exp{P}{v}$};
    \fill[darkgray] (cc) circle (2pt);
    \fill[darkgray] (ccc) circle (2pt);
    \node[] at (1.825, 2.5) {\fbox{$(P,\, v,\, q) \in \LEXP$}};
\end{tikzpicture}
\qquad
\begin{tikzpicture}[scale=1.2]
    \draw [<->,thick] (0,2.5) node (yaxis) [above] {}
        |- (3.5,0) node (xaxis) [right] {};
    \draw (0,0) coordinate (a_1) -- (0.5,.55) coordinate (a_2);
    \draw (.5,.55) coordinate (a_3) -- (1.5,1) coordinate (a_4);
    \draw (1.5,1) coordinate (a_5) -- (3.5,1.2) coordinate (a_6);
    \draw[dashed] (0, 1.9) coordinate (b_1) -- (3.5, 1.9) coordinate (b_2) node[left] at (0, 1.9) {$q$};
     \draw[dashdotted] (0, 1.3)  coordinate (e1) -- (3.5, 1.3) coordinate (e2);
    
    \coordinate (c1) at (0.5, 0);
    \coordinate (c2) at (0.5, 1);
    \coordinate (c) at (intersection of a_1--a_2 and c1--c2);
    \draw[dotted] (c) -- (xaxis -| c);
    \coordinate (cc1) at (1.5, 0);
    \coordinate (cc2) at (1.5, 1);
    \coordinate (cc) at (intersection of a_3--a_4 and cc1--cc2);
    \draw[dotted] (cc) -- (xaxis -| cc) node[below] {${\longleftarrow} ~ {\forall y_1,y_2} ~ {\longrightarrow}$};
    \coordinate (ccc1) at (2.5, 0);
    \coordinate (ccc2) at (2.5, 1);
    \coordinate (ccc) at (intersection of a_5--a_6 and ccc1--ccc2);
    \draw[dotted] (xaxis -| ccc) coordinate (dac) -- (ccc) coordinate (dacc);
    
    \draw[<->, thick] (2.5,1.3) -- (2.5, 1.9) node at (2.3, 1.6) {$\delta$};

     \draw[dashed] (ccc1) -- (.875,0);
    \fill[darkgray] (ccc) circle (2pt);
     \fill[darkgray] (cc) circle (2pt);
      \fill[darkgray] (c) circle (2pt);

   
    \node[left] at (0, 1.3) {$\Exp{P}{v}$};

     \node[] at (1.825, 2.5) {\fbox{$(P,\, v,\, q) \in \REXP$}};
\end{tikzpicture}
\end{center}
\caption{Schematic depiction of the formulae defining $\LEXP$ and $\REXP$, respectively. In each diagram, the solid line represents the monotonically increasing graph of $\sum_{0 \leq i \leq y_1}\: \sum_{0 \leq j \leq y_2} \wp\left(\Tps\left(\left\langle P,\, \eta_0,\, 1,\, \varepsilon\right\rangle,\, j,\, h(i)\right),\, v  \right)$ plotted over increasing $y_1$ and $y_2$.}
\label{fig}
\end{figure}

After establishing $\REXP \in \Sigma_2^0$ there is in principle still hope that $\REXP$ is recursively enumerable as $\Sigma_1^0 \subset \Sigma_2^0$.
We will, however, establish next that $\REXP \in \Sigma_2^0 \setminus \Pi_2^0 \not\supseteq \Sigma_1^0$ meaning that $\REXP$ is much harder to solve than, for instance, the halting problem.
To establish this, we will make use of a well--known $\Pi_2^0$--complete problem, namely the \emph{universal} halting problem for ordinary programs.
\begin{mydef}[The Universal Halting Problem]
The \textbf{universal halting problem} is a subset $\boldsymbol\UHP \subset \sonprogs$, which is characterized as follows:
\begin{align*}
P \in \UHP ~:\Longleftrightarrow~ \forall \eta \: \exists k \: \exists \eta'\colon \langle P,\, \eta,\, 1,\, \varepsilon\rangle \vdash^k \langle {\downarrow},\, \eta',\, 1,\, \varepsilon\rangle
\end{align*}
We denote by $\boldsymbol\cUHP$ the \textbf{complement of $\boldsymbol\UHP$}, i.e., $\cUHP = \sonprogs \setminus \UHP$.
\end{mydef}
In other words, a program $P$ is in $\UHP$, if it terminates its computation after a finite number of steps starting in \emph{any} initial valuation $\eta$. 
A characterization from a more computational point of view would be that $P \in \UHP$ if and only if $P$ satisfies $\forall \eta \: \exists k \: \exists \eta' \colon \Tps\big(\langle P,\, \eta,\, 1,\, \varepsilon\rangle,\, k,\, \varepsilon\big) = \langle {\downarrow},\, \eta',\, 1,\, \varepsilon\rangle$~.

The universal halting problem and its complement satisfy the following completeness properties:
\begin{mythm}[\cite{odifreddi2}]
\label{UHPcomplete}
$\UHP$ is $\Pi_2^0$--complete and $\cUHP$ is $\Sigma_2^0$--complete.
\end{mythm}
Next we will exploit Theorem \ref{UHPcomplete} to establish the $\Sigma_2^0$--completeness of $\REXP$:
\begin{mylem}
\label{RisSigmacomp}
$\REXP$ is $\Sigma_2^0$--complete.
\end{mylem}
\begin{proof}
By Lemma \ref{RinSigma2} we have $\REXP \in \Sigma_2^0$, so it remains to show that $\REXP$ is $\Sigma_2^0$--hard:
We do this by proving $\cUHP \leqm \REXP$.
Consider the following function $f\colon \cUHP \leqm \REXP$: $f$ takes an ordinary program $Q \in \sonprogs$ as its input and returns the triple $(P,\, v,\, 1)$, where $v$ does not occur in $Q$ and $P \in \soprogs$ is the following probabilistic program:
\begin{alltt}
\(i\) := 0; \{continue := 0\} [0.5] \{continue := 1\};
while (continue \(\neq\) 0)\{
    \(i\) := \(i\) + 1;
    \{continue := 0\} [0.5] \{continue := 1\}
\};
\(s\) := 0; \{continue := 0\} [0.5] \{continue := 1\};
while (continue \(\neq\) 0)\{
    \(s\) := \(s\) + 1;
    \{continue := 0\} [0.5] \{continue := 1\}
\};
\(v\) := 0; \(TQ\)
\end{alltt}
$TQ$ is a program that computes $\wp \big(\Tps\big(\big\langle Q\texttt{;}\:v \texttt{ := 1},\, g_Q(i),\, 1,\, \varepsilon \big\rangle,\, s,\, \varepsilon\big),\, v\big) \cdot 2^{s+1}$ and stores the result in the variable $v$, and $g_Q\colon \bbbn \rightarrow \mathbb V$ is some computable bijection, such that $\forall z \in \sovars\colon \big[g_Q(i)\big](z) \neq 0$ implies that $z$ occurs in $Q$.

\textit{\underline{Partial Correctness:}}  $\wp \big(\Tps\big(\big\langle Q\texttt{;}\:v \texttt{ := 1},\, g_Q(i),\, 1,\, \varepsilon \big\rangle,\, s,\, \varepsilon\big),\, v\big) \cdot 2^{s+1}$ returns $2^{s+1}$ if and only if $Q$ halts on input $g_Q(i)$ after exactly $s$ steps (otherwise 0), because only then, the variable $v$ is set to 1 after executing the program $Q\texttt{;}\:v \texttt{ := 1}$ for $s$ steps.
 The two while--loops generate independent geometric distributions with parameter $0.5$ on $i$ and $s$, respectively, so the probability of generating exactly the numbers $i$ and $s$ is $(2^i \cdot 2^s)^{-1}$. The expected value of $v$ after executing the program $P$ is hence
\begin{align*}
&\sum_{i\in\bbbn} \sum_{s\in\bbbn} \frac{1}{2^i \cdot 2^s} \cdot \wp\bigg(\Tps\Big(\big\langle Q\texttt{;}\:v \texttt{ := 1},\, g_Q(i),\, 1,\, \varepsilon \big\rangle,\, s + 1,\, \varepsilon\Big),\, v\bigg) \cdot 2^{s+1}~.
\end{align*}
Since for each input, the number of steps until termination is either unique or does not exist, the formula for the expected outcome reduces to $\sum_{i\in\bbbn} 2^{-i} \cdot 2= 1$ if and only if $Q$ halts on every input after some finite number of steps.
Thus if there exists an input on which $Q$  \emph{does not} eventually halt, then $(P,\, v,\, 1) \in \REXP$ as then the expected value is strictly less than one.
If, on the other hand, $Q$ \emph{does} halt on every input, then the expected value is exactly one and hence $(P,\, v,\, 1) \not\in \REXP$.

\textit{\underline{Total Correctness:}}
It is an easy but tedious exercise to construct a program computing $g_Q(i)$ given only $Q$. 
Program code for $\wp$, $\Tps$, multiplication and potentiation is also computable.
So in total, the program code for $P$  and thereby the triple $(P,\, v,\, 1)$ is computable.

By Theorem \ref{UHPcomplete}, $\cUHP$ is $\Sigma_2^0$--complete, so for any $\Problem{A} \in \Sigma_2^0$ it holds that $\Problem A \leqm \cUHP$.
Since we have just proven that $\cUHP \leqm \REXP$, it follows that $\Problem A \leqm \cUHP \leqm \REXP$, and by transitivity $\Problem A \leqm \REXP$.
\end{proof}

Finally, it follows from Lemma \ref{lemcompleteness} that membership for $\REXP$ is in some sense the hardest problem in $\Sigma_2^0$:
\begin{mycor}
\label{RinPinotSigma}
$\REXP \in \Sigma_2^0 \setminus \Pi_2^0$~.
\end{mycor}
\begin{myrem}
In the probabilistic program $P$ used in the proof of Lemma \ref{RisSigmacomp}, the randomization and the actual computation are completely separated.
Since $P$ in a sense ``solves" the universal halting problem, we interpret this as possible evidence that $P$ is in some sort of ``normal form" for probabilistic programs.
\end{myrem}

We now go on with characterizing the complexity of the set $\EXP$, which is the set we are mainly interested in, when it comes to expected outcomes.
As a first result we establish the following:
\begin{mylem}
\label{EinPi2}
$\EXP \in \Pi_2^0$~.
\end{mylem}
\begin{proof}
By Lemma \ref{RinSigma2}, there exists a decidable relation $\Problem U$, such that $(P,\, v,\, x) \in \REXP$ iff $\exists r_1 \forall r_2 \colon (r_1,\, r_2,\, P,\, v,\, x) \in \Problem U$.
Furthermore from Lemma \ref{LinSigma} it follows that there exists a decidable relation $\Problem L$, such that $(P,\, v,\, x) \in \LEXP \Longleftrightarrow \exists \ell \colon (\ell,\, P,\, v,\, x) \in \Problem L$.
Let $\neg \Problem U$ and $\neg \Problem L$ be the (decidable) negations of $\Problem U$ and $\Problem L$, respectively, then:
\belowdisplayskip=-1\baselineskip
\begin{align*}
~&(P,\, v,\, q) \in \EXP\\
\Longleftrightarrow~&q = \Exp{P}{v}\\
\Longleftrightarrow~&q \leq \Exp{P}{v} \,\wedge\, q \geq \Exp{P}{v}\\
\Longleftrightarrow~&\neg \big(q > \Exp{P}{v}\big) \,\wedge\, \neg \big(q < \Exp{P}{v}\big)\\
\Longleftrightarrow~&\neg \big(\exists r_1\: \forall r_2 \colon (r_1,\, r_2,\, P,\, v,\, q) \in \Problem U\big) \,\wedge\, \neg \big(\exists \ell \colon (\ell,\, P,\, v,\, q) \in \Problem L\big)\\
\Longleftrightarrow~&\big(\forall r_1\: \exists r_2 \colon (r_1,\, r_2,\, P,\, v,\, q) \in \neg \Problem U\big) \,\wedge\, \big(\forall \ell \colon (\ell,\, P,\, v,\, q) \in \neg \Problem L\big)\\
\Longleftrightarrow~&\forall r_1\: \forall \ell\: \exists r_2 \colon (r_1,\, r_2,\, P,\, v,\, q) \in \neg \Problem U \,\wedge\, (\ell,\, P,\, v,\, q) \in \neg \Problem L\\
\Longrightarrow~&\EXP \in \Pi_2^0
\end{align*}
\end{proof}

Intuitively the above proof asserts that we check whether $q = \Exp{P}{v}$ by deciding both $q \leq \Exp{P}{v}$ and $q \geq \Exp{P}{v}$ and that this check can be done by deciding a $\Pi_2^0$--relation.
Furthermore, we now establish the main theorem showing that $\EXP$ is $\Pi_2^0$--complete, thus extremely hard to solve:
\begin{mythm}
\label{EisPicomp}
$\EXP$ is $\Pi_2^0$--complete.
\end{mythm}
\begin{proof}
By Lemma \ref{EinPi2}, $\EXP \in \Pi_2^0$, so it remains to show that $\EXP$ is $\Pi_2^0$--hard.
We do this by proving $\UHP \leqm \EXP$.
Consider again the function $f$ from the proof of Lemma \ref{RisSigmacomp}: Given an ordinary program $Q$, $f$ computes the triple $(P,\, v,\, 1)$, where $P$ is a probabilistic program $P$ which has an expected outcome of one for the variable $v$ if and only if $Q$ terminates on all inputs, which is nothing else than $Q \in \UHP$.
Thus $f\colon \UHP \leqm \EXP$.

By Theorem \ref{UHPcomplete}, $\UHP$ is $\Pi_2^0$--complete, so for any $\Problem{A} \in \Pi_2^0$ it holds that $\Problem A \leqm \UHP$.
Since we have just proven that $\UHP \leqm \EXP$, it follows that $\Problem A \leqm \UHP \leqm \EXP$, and by transitivity $\Problem A \leqm \EXP$.
\end{proof}

It now follows from Lemma \ref{lemcompleteness} that membership for $\EXP$ is in some sense the hardest problem in $\Pi_2^0$:
\needspace{2\baselineskip}
\begin{mycor}
\label{EinPinotSigma}
$\EXP \in \Pi_2^0 \setminus \Sigma_2^0$~.
\end{mycor}

\section{Hardness of Deciding Almost--Sure Termination}
In this section, we turn towards the problem of almost--sure termination and establish completeness results for this problem.
We first establish that almost--sure termination is many--one reducible to $\EXP$ and thereby lays in $\Pi_2^0$:
\begin{mylem}
\label{ASTmEXP}
$\AST \leqm \EXP$~.
\end{mylem}

\begin{proof}
Consider the following function $f$ which takes a probabilistic program $Q$ as its input and returns the triple $(P,\, v,\, 1)$, where $P$ is the following probabilistic program:
\begin{alltt}\(v\) := 0; \(Q\); \(v\) := 1 \end{alltt}
\textit{\underline{Total and Partial Correctness:}} The triple $(P,\, v,\, 1)$ is obviously computable.
On executing $P$, the variable $v$ is set to one only in those runs in which the program $Q$ terminates. So the expected value of $v$ converges to one, if and only if the probability of $Q$ terminating converges to one.
So if $Q \in \AST$, then and only then $(P,\, v,\, 1) \in \EXP$.
Thus $f\colon \AST \leqm \EXP$~.
\end{proof}

By Theorem \ref{EisPicomp}, $\EXP$ is $\Pi_2^0$--complete, so it follows directly from Lemma \ref{ASTmEXP} that:
\begin{mycor}
\label{ASTinPi}
$\AST \in \Pi_2^0$~.
\end{mycor}
Next we will establish the $\Pi_2^0$--hardness and thereby the $\Pi_2^0$--completeness of $\AST$ by many--one--reduction from the universal halting problem:
\begin{mythm}
\label{ASTisPicomp}
$\AST$ is $\Pi_2^0$--complete.
\end{mythm}
\begin{proof}
By Corollary \ref{ASTinPi}, $\AST \in \Pi_2^0$, so it remains to show that $\AST$ is $\Pi_2^0$--hard.
For that we many--one reduce the $\Pi_2^0$--complete universal halting problem to $\AST$ using the following function $f\colon \UHP \leqm \AST$: $f$ takes an ordinary program $Q$ as its input and returns the following probabilistic program $P$:
\begin{alltt}
\(i\) := 0; \{continue := 0\} [0.5] \{continue := 1\};
while (continue \(\neq\) 0)\{
    \(i\) := \(i\) + 1;
    \{continue := 0\} [0.5] \{continue := 1\}
\};
\(TQ\)
\end{alltt}
where $TQ$ is an ordinary program that simulates the program $Q$ on input $g_Q(i)$, and $g_Q\colon \bbbn \rightarrow \mathbb V$ is some computable bijection, such that $\forall v \in \sovars\colon \big[g_Q(i)\big](v) \neq 0$ implies that $v$ occurs in $Q$.

\textit{\underline{Partial Correctness:}}
The while--loop in $P$ establishes a geometric distribution with parameter $0.5$ on $i$ and hence a geometric distribution on all possible inputs for $Q$.
After the while--loop, the program $Q$ is simulated on the input generated probabilistically in the while--loop.
Obviously then the entire program $P$ terminates with probability one, i.e., terminates almost--surely, if and only if the simulation of $Q$ terminates on every input.
Thus $Q \in \UHP$ if and only if $P \in \AST$.

\textit{\underline{Total Correctness:}} As mentioned in the proof of Lemma \ref{RisSigmacomp}, the program code for $g_Q$ is computable.
Also the program code for a universal program capable of simulating any program $Q$ on a given input is computable \cite{kleeneNF}.
So in total, the program code for $P$ is computable.

By Theorem \ref{UHPcomplete}, $\UHP$ is $\Pi_2^0$--complete, so for any $\Problem{A} \in \Pi_2^0$ it holds that $\Problem A \leqm \UHP$.
Since we have just proven that $\UHP \leqm \AST$, it follows that $\Problem A \leqm \UHP \leqm \AST$, and by transitivity $\Problem A \leqm \AST$. 
\end{proof}

\section{Conclusion}
Our results show that one can effectively enumerate all rationals that are strictly less than the expected outcome for a program variable $v$ after executing a probabilistic program $P$, i.e., arbitrarily close approximations from below are computable.
Obtaining such approximations from above is harder:
These would be recursively enumerable only if there would be access to an oracle for the (non--universal) halting problem \cite{kleeneNF,odifreddi1}.
This approximation problem is as hard to solve as deciding whether e.g., an ordinary program halts on finitely many inputs \cite{odifreddi2}.

Deciding almost--sure termination is even harder and is as hard as computing exact expected outcomes.
Namely, such outcomes are not recursively enumerable even if there would be access to an oracle for the halting problem \cite{odifreddi2}.
Other natural examples that are equally hard are the universal halting problem and the problem of deciding whether an ordinary program halts on infinitely many inputs \cite{odifreddi2}.

The established hardness results give insights into the specific difficulties of dealing with the studied decision problems.
In particular further research could be directed towards identifying subsets of probabilistic programs for which the upper bounds of the expected outcome are given by a $\Sigma_2^0$--set $\Problem A = \{x ~|~ \exists y_1 \forall y_2\colon (x,\, y_1,\, y_2) \in R\}$ such that the set $\Problem A' = \{(x, y_1) ~|~ \forall y_2\colon (x,\, y_1,\, y_2) \in R\}$ is decidable.
In this case, the set $\Problem A$ of upper bounds would be recursively enumerable and thus the \emph{exact} expected outcome can be approximated arbitrarily close from below \emph{and} from above.
Obtaining and deciding $\Problem A'$ would basically amount to transforming a given probabilistic program into an ordinary program for which then a non--termination proof has to be found which in certain cases can be automated.

Aside from the above considerations the structure of the probabilistic programs we use in our proofs hints towards the possible existence of a normal form for probabilistic programs in which the randomization and the actual computation are separated. 
Further investigation of this issue is planned.
Further research could also deal with hardness of the considered problems in presence of non--determinism.



\bibliographystyle{plain}
\bibliography{literature}

\end{document}